\documentclass[12pt,reqno]{article}

\usepackage[usenames]{color}
\usepackage{amssymb}
\usepackage{graphicx}
\usepackage{amscd}
\usepackage{pifont}
\usepackage{appendix}

\usepackage[colorlinks=true,
linkcolor=webgreen,
filecolor=webbrown,
citecolor=webgreen]{hyperref}

\definecolor{webgreen}{rgb}{0,.5,0}
\definecolor{webbrown}{rgb}{.6,0,0}

\usepackage{fullpage}
\usepackage{float}
\usepackage[shortlabels]{enumitem}

\usepackage{graphics,amsmath,amssymb}
\usepackage{amsthm}
\usepackage{amsfonts}
\usepackage{latexsym}
\usepackage{epsf}
\usepackage{mathtools}
\usepackage{makecell}

\usepackage{tikz}
\usetikzlibrary{automata, arrows.meta, positioning}

\DeclarePairedDelimiter{\abs}{\lvert}{\rvert}

\newcommand{\seqnum}[1]{\href{http://oeis.org/#1}{\underline{#1}}}

\def\modd#1 #2{#1\ \mbox{\rm (mod}\ #2\mbox{\rm )}}

\newcommand{\xmark}{\text{\ding{55}}}

\def\suchthat{\, : \,}

\begin{document}

\theoremstyle{plain}
\newtheorem{theorem}{Theorem}
\newtheorem{corollary}[theorem]{Corollary}
\newtheorem{lemma}[theorem]{Lemma}
\newtheorem{proposition}[theorem]{Proposition}

\theoremstyle{definition}
\newtheorem{definition}[theorem]{Definition}
\newtheorem{example}[theorem]{Example}
\newtheorem{conjecture}[theorem]{Conjecture}

\theoremstyle{remark}
\newtheorem{remark}[theorem]{Remark}

\title{Strongly $k$-recursive sequences}

\author{Daniel Krenn\\
Fachbereich Mathematik\\
Paris Lodron University of Salzburg\\
Hellbrunner Stra{\ss}e 34\\
5020 Salzburg\\
Austria\\
\href{mailto:math@danielkrenn.at}{\tt math@danielkrenn.at} \\
or \href{mailto:daniel.krenn@plus.ac.at}{\tt daniel.krenn@plus.ac.at} 
\and
Jeffrey Shallit\thanks{Supported by NSERC Grant 2018-04118.} \\
School of Computer Science\\
University of Waterloo\\
200 University Ave. W.\\
Waterloo, ON  N2L 3G1 \\
Canada \\
\href{mailto:shallit@uwaterloo.ca}{\tt shallit@cs.uwaterloo.ca} }
\date{}

\maketitle

\begin{abstract}
Drawing inspiration from a recent paper of Heuberger, Krenn, and Lipnik,
we define the class of strongly $k$-recursive sequences. We
show that every $k$-automatic sequence is strongly $k$-recursive,
therefore $k$-recursive, and discuss that the converse is not true.

We also show that the class of strongly $k$-recursive sequences is
a proper subclass of the class of $k$-regular sequences, and we present some
explicit examples.  We then extend the proof techniques to answer the same question for the class of $k$-recursive sequences.
\end{abstract}

\section{Introduction}

\paragraph{Four classes of sequences.}
Here are three well-studied interesting classes of sequences related to base-$k$ expansions.

\begin{enumerate}
\item
A sequence $(f(n))_{n \geq 0}$ is said to be {\it $k$-automatic\/} if there exists
a deterministic finite automaton with output that, on input
$n$ expressed in base $k$, reaches a state with output $f(n)$;
see~\cite{Allouche&Shallit:2003a} for more details.
An example is the Thue--Morse sequence~$\bf t$;
we will use it in Example~\ref{example:factor-compl-tm}.

\item
A sequence $(f(n))_{n \geq 0}$ is said to be 
{\it $k$-regular\/}
if there exists
a finite set $S$ of subsequences
of the form $(f(k^r n + a))_{n \geq 0}$ 
with $0 \leq a < k^r$, and
such that every subsequence of the form
$(f(k^t n + b))_{n \geq 0}$ with $t \geq 0$ and $0 \leq b < k^t$
is a linear combination of
the elements of $S$.  This is an interesting class that
is discussed, for example,
in~\cite{Allouche&Shallit:1992,Allouche&Shallit:2003a}.  A classic example of a $k$-regular sequence is $(s_k (n))_{n \geq 0}$, where
$s_k(n)$ is the sum of the digits of $n$ when expressed in base~$k$.

\item
A sequence $(f(n))_{n \geq 0}$ taking its values in the natural numbers is said
to be {\it $k$-synchronized\/} if there exists a deterministic finite automaton that
accepts $\{ (n,m)_k \suchthat m = f(n) \}$.
Here $(n,m)_k$ is a representation of
the pair $(n,m)$ in base~$k$, where we pad the shorter representation
with leading zeros, if necessary, and the digits of the two numbers
are read in parallel. More about synchronized sequences can be found,
for example, in~\cite{Shallit:2021h}.
\end{enumerate}

It is known that every $k$-automatic sequence is $k$-synchronized (see \cite[Theorem~4]{Shallit:2021h}), and every
$k$-synchronized sequence is $k$-regular
(see \cite[Theorem~6]{Shallit:2021h}).

Recently Heuberger, Krenn, and Lipnik~\cite{Heuberger&Krenn&Lipnik:2022}
introduced a fourth related class of sequences, the
$k$-recursive sequences.

\begin{enumerate}[resume]
\item
We say a sequence $(f(n))_{n \geq 0}$ is {\it $k$-recursive\/}
if there exist two natural numbers $r < t$ and two integers $L < U$ and a natural number $n_0$ such that
every subsequence of the form
$(f(k^t n + b))_{n \geq n_0}$ with $0 \leq b < k^t$
is a linear combination of the elements of the set
$$ \{ (f(k^r + a))_{n \geq n_0} \suchthat L \leq a < U  \} .$$
\end{enumerate}

The authors of~\cite{Heuberger&Krenn&Lipnik:2022} proved that every $k$-recursive sequence is $k$-regular.
This suggests the question of whether there is a $k$-regular
sequence that is not $k$-recursive. In this paper,
we provide an affirmative answer (Theorem~\ref{theorem:g-not-recursive}).

\paragraph{Strongly $k$-recursive sequences.}
The investigation of the question above leads to a
natural variation on the definition in~\cite{Heuberger&Krenn&Lipnik:2022} and
studying the corresponding class of sequences:
in the definition of a $k$-recursive sequence,
we insist that $n_0 = 0$, $L \geq 0$, and $U \leq k^t$.  We call such a sequence
{\it strongly $k$-recursive}.

Clearly, every strongly $k$-recursive sequence is $k$-recursive. In
this paper we determine the relationship of this new class to the
classes of automatic (Theorem~\ref{theorem:automatic-is-strongly-recursive}),
synchronized (Proposition~\ref{proposition:g-is-synchronized}), and
regular sequences (Theorem~\ref{theorem:g-not-strongly-recursive}).

Before presenting and proving the actual results, we start by an
example of a strongly $k$-recursive sequence.

\begin{example}\label{example:factor-compl-tm}
Here is an example of a strongly $2$-recursive sequence.
Let $f_{\bf t}(n)$ denote the factor complexity of the Thue--Morse
sequence; that is, the number of distinct length-$n$ blocks occurring
in the sequence ${\bf t} = 0110100110010110\dots$;
see~\cite{Avgustinovich:1994,Brlek:1989}.  Then we claim that
\begin{align*}
f_{\bf t}(16n) &= f_{\bf t}(8n)-f_{\bf t}(8n+1)+3f_{\bf t}(8n+2)-f_{\bf t}(8n+4) \\
f_{\bf t}(16n+1) &= 3f_{\bf t}(8n+2)-f_{\bf t}(8n+4)\\
f_{\bf t}(16n+2) &= -f_{\bf t}(8n+1)+4f_{\bf t}(8n+2)-f_{\bf t}(8n+4)\\
f_{\bf t}(16n+3) &= -2f_{\bf t}(8n+1) + 5f_{\bf t}(8n+2) - f_{\bf t}(8n+4)\\
f_{\bf t}(16n+4) &= -f_{\bf t}(8n+1) + 3f_{\bf t}(8n+2)\\
f_{\bf t}(16n+5) &= -2f_{\bf t}(8n+1) + 4f_{\bf t}(8n+2)\\
f_{\bf t}(16n+6) &= -f_{\bf t}(8n+1) + 2f_{\bf t}(8n+2) + f_{\bf t}(8n+4)\\
f_{\bf t}(16n+7) &= 2f_{\bf t}(8n+4)\\
f_{\bf t}(16n+8) &= -f_{\bf t}(8n+1)+f_{\bf t}(8n+2)+2f_{\bf t}(8n+4)\\
f_{\bf t}(16n+9) &=  -2f_{\bf t}(8n+1)+2f_{\bf t}(8n+2)+2f_{\bf t}(8n+4)\\
f_{\bf t}(16n+10) &=  -f_{\bf t}(8n+1)+f_{\bf t}(8n+2)+f_{\bf t}(8n+4)+f_{\bf t}(8n+6)\\
f_{\bf t}(16n+11) &= 2f_{\bf t}(8n+6)\\
f_{\bf t}(16n+12) &= f_{\bf t}(8n+1)-f_{\bf t}(8n+2)-f_{\bf t}(8n+4)+3f_{\bf t}(8n+6)\\
f_{\bf t}(16n+13) &= 2f_{\bf t}(8n+1)-2f_{\bf t}(8n+2)-2f_{\bf t}(8n+4)+4f_{\bf t}(8n+6)\\
f_{\bf t}(16n+14) &= f_{\bf t}(8n+1)-f_{\bf t}(8n+4)+5f_{\bf t}(8n+6)\\
f_{\bf t}(16n+15) &= 2f_{\bf t}(8n+2)-6f_{\bf t}(8n+4)+6f_{\bf t}(8n+6),
\end{align*}
which shows that $(f_{\bf t}(n))_{n \geq 0}$ is strongly $2$-recursive.
These relations can be easily proved using
the {\tt Walnut} theorem prover~\cite{Mousavi:2016,Shallit:2022},
and the necessary commands appear
in Appendix~\ref{sec:walnut-code}.
\label{examp1}
\end{example}

\section{Automatic sequences}

We begin our studies by determining the relationship between the
$k$-automatic and strongly $k$-recursive sequences.

\begin{theorem}\label{theorem:automatic-is-strongly-recursive}
Every $k$-automatic sequence is strongly $k$-recursive.
\end{theorem}
As every strongly $k$-recursive sequence is $k$-recursive, we conclude that
every $k$-automatic sequence is also $k$-recursive.
The converse, however, is not true, as there are $k$-recursive and strongly
$k$-recursive sequences that are not bounded
(see~\cite{Heuberger&Krenn&Lipnik:2022}), but automatic sequences are
bounded (see~\cite{Allouche&Shallit:2003a}).

Theorem~\ref{theorem:automatic-is-strongly-recursive}
is trivially implied by the following stronger result.

\begin{theorem}
For every $k$-automatic sequence $(f(n))_{n\geq 0}$
there exist natural numbers $r<t$ such that
every subsequence of the form
$(f(k^t n + b))_{n \geq 0}$ with $0 \leq b < k^t$ is equal to some subsequence
of the form $(f(k^r n + a))_{n \geq 0}$ with $0 \leq a < k^r$.
\end{theorem}

\begin{proof}
This follows from Theorem 1 of Cobham~\cite{Cobham:1972}, but
for completeness we give the following simpler proof.

Let the $k$-automatic sequence $(f(n))_{n \geq 0}$ be generated by the
least significant digit-first deterministic finite automaton with output $(Q, \Sigma_k, \Delta, \delta, q_0, \tau)$, where
$\Sigma_k = \{0,1,\ldots, k-1 \}$.  

For each
natural number $t$ consider the set%
\footnote{We write $\abs{x}$ for the length of a word $x \in \Sigma_k^*$.}
$$ S_t := \{ q \in Q \suchthat \exists x \in \Sigma_k^*,\ \abs{x} = t \text{ and } \delta(q_0, x) = q \}.$$
Then there are only $2^{\abs{Q}}$ distinct sets $S_t$, so there must exist
natural numbers $r < t$ such that
$S_t$ is a subset of $S_r$ .

Now consider the subsequence
$(f(k^t n + b))_{n \geq 0}$ for some integer~$b$ with $0 \leq b < k^t$.
Then there exists an~$x \in \Sigma_k^*$ with $\abs{x}=t$ and $x$ is
the standard base-$k$ representation of $b$.
Therefore, there is a $q \in S_t$ with $\delta(q_0, x) = q$. Consequently,
we have $q \in S_r$, so there is a~$y \in \Sigma_k^*$ with $\abs{y}=r$ and
$\delta(q_0, y) = q$. Let $a$ be the integer that $y$ represents in base~$k$.
Thus, for all $n \ge 0$, we conclude that $f(k^r n + a) = f(k^r n + b)$
by first reading either $x$ (representation of~$a$) or $y$ (representation of~$b$)
and then the representation of~$n$ in base~$k$.
Therefore, we have that every subsequence of the form
$(f(k^t n + b))_{n \geq 0}$, $0 \leq b < k^t$, is equal to some subsequence
of the form $(f(k^r n + a))_{n \geq 0}$, $0 \leq a < k^r$.
\end{proof}

\section{A counterexample sequence}

Before we actually study the relationship of the classes of
$k$-recursive and strongly $k$-recursive sequences to other classes,
we put our preparatory focus on a particular set of sequences (defined below).
We will derive some basic properties of these sequences---in particular,
that these sequences are $k$-regular. Moreover, we study when these sequences are
$k$-synchronized and therefore providing answers of some of the
introductory questions.

Throughout this section, let $k$, $\ell \geq 2$ be integers and define
$$ g_{k, \ell}(n)  = \begin{cases}
    1, & \text{if $n = 0$}; \\
    1+\ell^{\lfloor \log_k n \rfloor}, & \text{if $n > 0$.}
\end{cases}
$$
Note that $g_{k, \ell}(1) = 2$ (independently of $k$ and $\ell$).

We will begin our studies by the following properties stated in the
two lemmas below.
\begin{lemma}
Let $n \geq 0$ and $t \geq 0$. We have
$$ g_{k, \ell}(k^t n + b) = (1-\ell^j)g_{k, \ell}(k^t n) + \ell^j g_{k, \ell}(k^t n + 1)$$
for $k^j \leq b < k^{j+1}$ with $0 \leq j < t$.
\label{lem5}
\end{lemma}

\begin{proof}
Suppose $n \geq 1$.  Let $s \geq 0$ be such that $k^s \leq n < k^{s+1}$. We have
$1 \leq b < k^t$; therefore, for each such $b$ as well as $b=0$, we obtain
$\lfloor \log_k (k^t n + b) \rfloor = t+s$,
so $g_{k, \ell}(k^t n + b) = 1+ \ell^{t+s}$.
In particular, we get $g_{k, \ell}(k^t n) = g_{k, \ell}(k^t n + 1) = g_{k, \ell}(k^t n + b)$,
and the result follows.

Suppose $n = 0$.   Then $\lfloor \log_k (k^t n + b) \rfloor = \lfloor \log_k (b) \rfloor = j$,
so $g_{k, \ell}(k^t n + b) = g_{k, \ell}(b) = 1+ \ell^j$,
while $g_{k, \ell}(k^t n) = g_{k, \ell}(0) = 1$ and $g_{k, \ell}(k^t n + 1) = g_{k, \ell}(1) = 2$, and
the result follows.
\end{proof}

\begin{lemma}\label{lem:g-of-k2n}
Let $n \geq 0$.
We have
\begin{subequations}
\begin{align}
g_{k,\ell} (k^2 n) &= -\ell g_{k, \ell} (n) + (\ell + 1) g_{k,\ell} (kn),  \label{eq1}\\
g_{k,\ell} (k^2 n + a) &= - \ell g_{k, \ell} (n) + \ell g_{k,\ell} (kn) + g_{k,\ell} (kn+1) \quad \text{for $1 \leq a < k$},  \label{eq2}\\
g_{k,\ell} (k^2n + b) &= -\ell g_{k, \ell} (n) + g_{k, \ell} (kn) + \ell g_{k,\ell} (kn+1) \quad \text{for $k \leq b < k^2$}. \label{eq3}
\end{align}
\end{subequations}
\end{lemma}

\begin{proof}
\leavevmode 
\ \vphantom{a} \ 

\eqref{eq1}:  If $n = 0$, then both sides are equal to $1$.
Otherwise, let $s \geq 0$ be such that $k^s \leq n < k^{s+1}$.
Then the left-hand side $g_{k,\ell} (k^2 n)$ evaluates to $1 + \ell^{s+2}$,
while $g_{k, \ell}(n) = 1+\ell^s$ and
$g_{k, \ell}(kn) = 1+\ell^{s+1}$.  

\eqref{eq2}:  If $n=0$, then both sides are equal to $2$.
Otherwise, let $s \geq 0$ be such that $k^s \leq n < k^{s+1}$.
Then the left-hand side $g_{k,\ell}(k^2n+a)$ evaluates to $1+\ell^{s+2}$,
while the right-hand side evaluates to
$-\ell(1+ \ell^s) + \ell(1+\ell^{s+1}) + (1 + \ell^{s+1}) = 1+\ell^{s+2}$.

\eqref{eq3}:  If $n=0$, then both sides equal $1+\ell$.
Otherwise, let $s \geq 0$ be such that $k^s \leq n < k^{s+1}$.
Then the left-hand side $g_{k,\ell}(k^2n+b)$ evaluates to $ 1+\ell^{s+2}$,
while the right-hand side evaluates to
$  -\ell(1+ \ell^s) + (1+\ell^{s+1}) + \ell(1+\ell^{s+1}) = 1 + \ell^{s+2}$.  
\end{proof}

We are now ready to prove the following proposition.

\begin{proposition}\label{proposition:g-regular}
The sequence $(g_{k,\ell} (n))_{n \geq 0}$ is $k$-regular.  
\end{proposition}

\begin{proof}
We show that every subsequence of the form
$(g(k^t n + b))_{n \geq 0}$ with $0 \leq b < k^t$
is a linear combination of the three sequences
$(g(n))_{n \geq 0}$, $(g(kn))_{n \geq 0}$, and
$(g(kn+1))_{n \geq 0}$ by an induction on~$t$.
We use Lemma~\ref{lem:g-of-k2n} as a base.

Now, let $t > 2$ and consider $(g(k^t n + b))_{n \geq 0}$ with $0 \leq b < k^t$.
By Lemma~\ref{lem5}, this sequence is a linear combination of
$(g(k^t n))_{n \geq 0}$ and
$(g(k^t n + 1))_{n \geq 0}$.
First splitting $k^tn = k^2(k^{t-2})n$, then applying Lemma~\ref{lem:g-of-k2n}---specifically we use relations \eqref{eq1} and \eqref{eq2}---and then the induction hypothesis, completes the proof.
\end{proof}

Let us now focus on $k$-synchronized sequences. We show the following
proposition.
\begin{proposition}\label{proposition:g-is-synchronized}
The sequence $(g_{k,\ell}(n)))_{n \geq 0}$ is $k$-synchronized iff $k = \ell$.
\end{proposition}
We discuss the implications of this result in Section~\ref{sec:summary}.

\begin{proof}
Suppose $k = \ell$.
For $k = 2$ this is demonstrated by the
(most-significant-digit first) deterministic finite automaton in Figure~\ref{fig2}, and for
$k > 2$ by the automaton in Figure~\ref{figk}.  As usual, accepting states are denoted by double circles, and
in both cases unspecified transitions go to a ``dead
state'' that is non-accepting.

We distinguish two cases. First, the values $g_{k,k}(0)=1$ and $g_{k,k}(a)=2$ for $0<a<k$
are represented by the path from state $0$ via $3$ to $2$ in Figure~\ref{fig2} for $k=2$ and
directly from state $0$ to $2$ in Figure~\ref{figk} for $k>2$.
Second, all other values $g_{k,k}(n)$ are represented by
the path from state $0$ via $1$ to $2$ in both Figures~\ref{fig2} and~\ref{figk}.
Note that for any $n$ whose base-$k$ representation is of length~$L$
we have that the base-$k$ representation of $g_{k,k}(n)$ is
$10^{L-2}1$.

\begin{figure}[H]
\begin{center}
    \begin{tikzpicture} [scale=2, >=latex, thick, every node/.style={align=center}]
    \node (q0) at (0,0) [state, initial, initial text = {}] {$0$};
    \node (q1) at (3,0.5) [state] {$1$};
    \node (q2) at (6,0) [state, accepting] {$2$};
    \node (q3) at (3,-0.5) [state, accepting] {$3$};
 
    \path  [-latex]
        (q0) edge [loop above] node {$[0,0]$} ()
        (q0) edge node [above, sloped] {$[1,1]$} (q1)
        (q1) edge node [above, sloped] {$[0,1]$\\$[1,1]$}   (q2)
        (q1) edge [loop above] node {$[0,1]$\\$[1,1]$} ()
        (q0) edge node [below, sloped] {$[0,1]$} (q3)
        (q3) edge node [below, sloped] {$[1,0]$} (q2);
    \end{tikzpicture} 
\end{center}
\caption{$g_{2,2}(n)$ is $2$-synchronized.}
\label{fig2}
\end{figure}
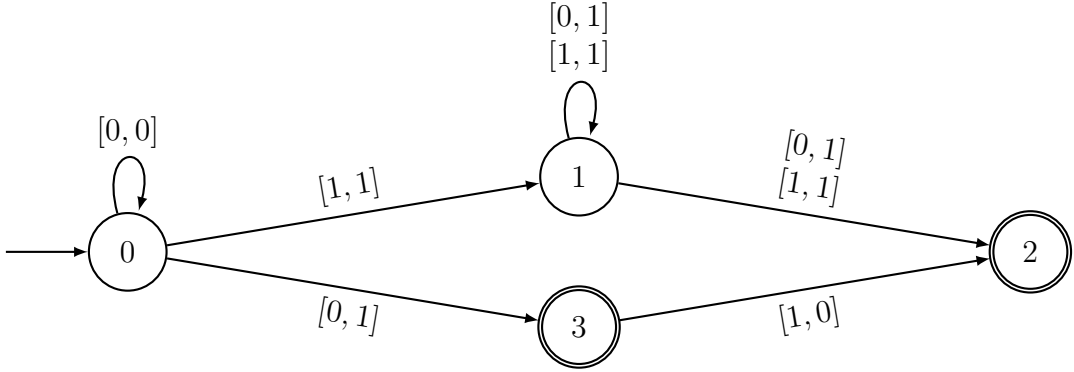

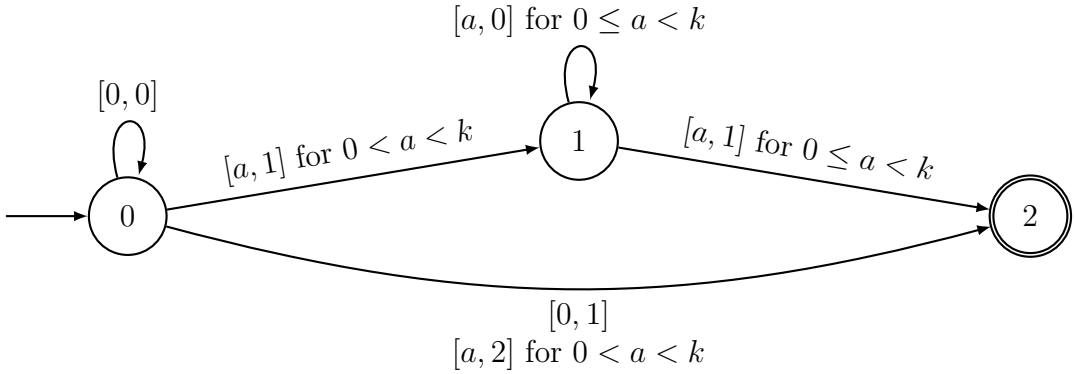
\begin{figure}[H]
\begin{center}
    \begin{tikzpicture} [scale=2, >=latex, thick, every node/.style={align=center}]
    \node (q0) at (0,0) [state, initial, initial text = {}] {$0$};
    \node (q1) at (3,0.5) [state] {$1$};
    \node (q2) at (6,0) [state, accepting] {$2$};
 
    \path  [-latex]
        (q0) edge [loop above] node {$[0,0]$} ()
        (q0) edge node [above, sloped] {$[a,1]$ for $0<a<k$} (q1)
        (q1) edge node [above, sloped] {$[a,1]$ for $0\leq a<k$}   (q2)
        (q1) edge [loop above] node {$[a,0]$ for $0\leq a<k$} ()
        (q0) edge [bend right=15] node [below] {$[0,1]$\\$[a,2]$ for $0<a<k$} (q2);
    \end{tikzpicture}
\end{center}
\caption{$g_{k,k}(n)$ is $k$-synchronized, $k > 2$.}
\label{figk}
\end{figure}

Now suppose $k \not= \ell$.
It is known that every $k$-synchronized
sequence is either $O(1)$, or $O(n)$ and $\geq cn$ infinitely often; see~\cite[Theorem~8]{Shallit:2021h}.
However
$g_{k, \ell} (n) = \Theta(n^{\log_k \ell})$,
so if $k \not=\ell$ then $g_{k,\ell} (n)$ cannot be $k$-synchronized.
\end{proof}

\section{Not (strongly) $k$-recursive}

We continue our study of the sequence~$(g_{k,\ell}(n))_{n \geq 0}$. In
this section, we will show that this sequence is neither strongly
$k$-recursive (Theorem~\ref{theorem:g-not-strongly-recursive}) nor
$k$-recursive (Theorem~\ref{theorem:g-not-recursive}).
Moreover, at the end, we bring another example of a $k$-regular sequence
that is not strongly $k$-recursive.

\begin{theorem}\label{theorem:g-not-strongly-recursive}
The sequence $(g_{k,\ell}(n))_{n \geq 0}$ is not
strongly $k$-recursive.
\end{theorem}

\begin{proof}
Let $r < t$ and assume, contrary to what we want to prove, that $(g_{k,\ell}(k^t n))_{n \geq 0}$ is a linear combination
of $(g_{k,\ell}(k^r n + a))_{n \geq 0}$ for $0 \leq a < k^r$.
Then from Lemma~\ref{lem5}
it follows that $(g_{k,\ell}(k^t n))_{n \geq 0}$ would be a linear
combination of $(g_{k,\ell}(k^r n))_{n \geq 0}$ and $(g_{k,\ell}(k^r n + 1))_{n \geq 0}$,
say $g_{k,\ell}(k^t n) = c_0 g_{k,\ell}(k^r n) + c_1 g_{k,\ell}(k^r n + 1)$ for all $n\geq0$.

By setting $n = 1$ and $n = k$ we get
the system
\begin{alignat*}{3}
    \ell^t + 1     &= c_0 (\ell^r + 1)     &&+ c_1 (\ell^r + 1)     &&= (c_0 + c_1) (\ell^r + 1), \\
    \ell^{t+1} + 1 &= c_0 (\ell^{r+1} + 1) &&+ c_1 (\ell^{r+1} + 1) &&= (c_0 + c_1) (\ell^{r+1} + 1).
\end{alignat*}
As $\ell^r+1>0$ and $\ell^{r+1}+1>0$, we have
\begin{equation*}
    \frac{\ell^t + 1}{\ell^r+1}
    = c_0 + c_1 =
    \frac{\ell^{t+1} + 1}{\ell^{r+1} + 1},
\end{equation*}
which implies $r = t$, a contradiction.
\end{proof}

We can modify and extend the technique in the proof of
Theorem~\ref{theorem:g-not-strongly-recursive}, so that we can also
show the corresponding result on $k$-recursive sequences.

\begin{theorem}\label{theorem:g-not-recursive}
The sequence $(g_{k,\ell}(n))_{n \geq 0}$ is not $k$-recursive.
\end{theorem}

\begin{proof}
Let $r < t$ and assume, contrary to what we want to prove, that there are integers $L < U$ such that $(g_{k,\ell}(k^t n))_{n \geq 0}$ is a linear combination
of $(g_{k,\ell}(k^r n + a))_{n \geq 0}$ for $L \leq a < U$, say
\begin{equation}\label{eq:not-rec:relation}
  g_{k,\ell}(k^t n) = \sum_{L \leq a < U} c_a g_{k,\ell}(k^r n + a)
  \quad\quad\text{for all $n \geq 0$.}
\end{equation}
We set $n=k^{s-1}(k+1)$ and choose $s$ large enough such that
$\lfloor \log_k (k^r n + a) \rfloor = \lfloor \log_k (k^r n) \rfloor = r+s$
for all $a$ with~$L \leq a < U$.
Note that any $s$ with $-k^{r+s-1} \leq L$ and $U \leq k^{r+s-1}$ works, because then
\begin{equation*}
    k^{r+s} = k^r k^{s-1}(k+1) - k^{r+s-1}
    \leq k^r n + a <
    k^r k^{s-1}(k+1) + k^{r+s-1} = k^{r+s} + 2k^{r+s-1} \leq k^{r+s+1}.
\end{equation*}
Therefore, from~\eqref{eq:not-rec:relation} and the definition of
$g_{k,\ell}(n)$, we deduce that
\begin{subequations}
\begin{equation}\label{eq:not-rec:subs-s}
    \ell^{t+s} + 1 = \sum_{L \leq a < U} c_a \bigl(\ell^{r+s} + 1\bigr)
    = \bigl(\ell^{r+s} + 1\bigr) c
\end{equation}
with $c=\sum_{L \leq a < U} c_a$.
Repeating with $n=k^{s}(k+1)$ (same $s$ as chosen above) yields
\begin{equation}\label{eq:not-rec:subs-s+1}
    \ell^{t+s+1} + 1 = \bigl(\ell^{r+s+1} + 1\bigr) c.
\end{equation}
\end{subequations}
Multiplying~\eqref{eq:not-rec:subs-s} by $\ell$ and subtracting~\eqref{eq:not-rec:subs-s+1} yields
$\ell - 1 = (\ell - 1) c$, which implies $c=1$.

On the other hand, subtracting~\eqref{eq:not-rec:subs-s} from~\eqref{eq:not-rec:subs-s+1} and inserting~$c=1$ yields
$(\ell-1)\ell^{t+s} = (\ell-1)\ell^{r+s} c = (\ell-1)\ell^{r+s}$.
This implies $r=t$, a contradiction.
\end{proof}

We round off this section by presenting another interesting sequence
that is $k$-regular but not strongly $k$-recursive (for the particular
case $k=3$).

\begin{remark}
Define the sequence $(h(n))_{n \geq 0}$ by%
\footnote{We write $n \bmod m$ for the remainder of the division of $n$ by $m$
  that is $\ge0$ and $<m$. To avoid confusion when reading, we will also use the usual
  $a \equiv \modd{b} {m}$ for $a$ and $b$ being congruent modulo~$m$.}
$$ h(n) = \begin{cases}
0, & \text{if $n = 0$}; \\
h(\lfloor n/3 \rfloor) + (\lfloor n/3 \rfloor \bmod 2), & \text{if 
$n \equiv \modd{0} {3}$ or $n \equiv \modd{2} {3}$}; \\
h(\lfloor n/9 \rfloor) + 1, & \text{if $n \equiv \modd{1} {3}$}.
\end{cases}
$$
This sequence was introduced in~\cite{Allouche&Shallit:2003b}, 
where it was conjectured that
$h(n) = \nu_3 (d(n))$.
Here,
$$ d(n) = \sum_{0 \leq k \leq n} {n \choose k} {{n+k} \choose k} $$
are the {\it central Delannoy numbers} (sequence 
\seqnum{A001850} in the On-Line Encyclopedia of Integer Sequences).
This conjecture was proven quite recently by Zhao Shen~\cite{Shen:2022}.

A proof that $(h(n))_{n \geq 0}$ is $3$-regular can
be found in~\cite{Shen:2022}.  It follows from the fact
that 
$h(3n+b) = h(n) + ((n+a) \bmod 2)$
for all $n\geq0$ and $b \in \{0,1,2\}$, easily proved by induction on $n$. 
From this another induction gives  
$$h(3^t n + b) = \begin{cases}
h(n)+h(b), & \text{if $n \equiv \modd{0} {2}$;} \\
h(n)+t-h(b), & \text{if $n \equiv \modd{1} {2}$.}
\end{cases}
$$
for all $t$, $n \geq 0$ and $0 \leq b < 3^t$.  
Hence $(h(n))_{n \geq 0}$ satisfies the relations
\begin{align*}
h(3n+2) &= h(n) + (n \bmod 2) = h(3n+2) \\
h(9n) &= h(9n+6) = h(n) + 2(n \bmod 2) = -h(n) + 2h(3n) \\
h(9n+1) &= h(9n+4) = h(9n+7) = h(n)+1 = -h(n) + h(3n) + h(3n+1) \\
h(9n+3) &= h(n)+2(1-(n \bmod 2)) = -h(n) + 2h(3n+1) .
\end{align*}

To see that  $(h(n))_{n \geq 0}$ is not strongly $3$-recursive, we follow a similar plan as we did for
$g_{k,\ell}$.  Namely, we show that for 
$0 \leq b < 3^t$,
each sequence of the form $(h(3^t n + b))_{n \geq 0}$ can
be expressed as a linear combination of
the two sequences $(h(3^t n))_{n \geq 0} $
and $(h(3^t n + 1))_{n \geq 0}$.
In fact, we have
\begin{equation}
h(3^t n + b) = 
(1-h(b)) h(3^t n) +
h(b) h(3^t n+1)
\label{fund}
\end{equation}
for all $t$, $n \geq 0$
and $0 \leq b < 3^t$.
Finally,
the sequence $(h(3^t n))_{n \geq 0}$ cannot be expressed
as a linear combination of the sequences $(h(3^r n + a))_{n \geq 0}$,
$0 \leq a < 3^r$, for $r < t$
because if it could,
by \eqref{fund} there would
be constants $c_0, c_1$ such that
$h(3^t n) = c_0 h(3^r n) + c_1 h(3^r n + 1)$ for all
$n \geq 0$.

Substituting $n \in \{0, 1, 3\}$ we get the system of equations
\begin{align*}
0 &= c_1 \\
t+1 &= c_0 (r+1) + c_1 r \\
t+2 &= c_0 (r+2) + c_1 (r+1),
\end{align*}
which forces $c_0 = 1$ and $r = t$, a contradiction.
\end{remark}

\section{Summary of results}
\label{sec:summary}

We now collect and discuss the results.
In Table~\ref{tab1}, we illustrates all the various
possibilities for $k$-regular sequences. 

\begin{table}[H]
    \centering
    \begin{tabular}{c|c|c|c|c|c}
    $k$-automatic & $k$-synchronized & \makecell{strongly \\ $k$-recursive} & $k$-recursive & $k$-regular & example \\
    \hline
    \xmark & \xmark & \xmark & \xmark & \checkmark & $g_{k,\ell}(n)$ with $k \not=\ell$ \\
     \xmark & \xmark & \checkmark & \checkmark & \checkmark    &  $s_k(n)$  \\
     \xmark & \checkmark & \xmark & \xmark & \checkmark & $g_{k,k}(n)$ \\
     \xmark & \checkmark & \checkmark & \checkmark & \checkmark & $n$ \\
     \checkmark & \checkmark & \checkmark & \checkmark & \checkmark & \makecell{every \\automatic sequence}
    \end{tabular}
    \caption{Possibilities for $k$-regular sequences}
    \label{tab1}
\end{table}

Note that the columns ``strongly $k$-recursive'' and ``$k$-recursive''
in the table are the same; only partial information on the
relationship between the corresponding classes of sequences is known.

\appendix
	
\section{{\tt Walnut} code}
\label{sec:walnut-code}

The {\tt Walnut} code below can be used to verify the claims in
Example~\ref{examp1}.   To download the latest
version of {\tt Walnut}, visit\\
\centerline{\url{https://cs.uwaterloo.ca/~shallit/walnut.html} \ .}

\begin{verbatim}
reg power2 msd_2 "0*10*":

def tmsc "(n<=1 & z=n+1) |  Ex,r $power2(x) & x<n & n<=2*x & n=x+r &
   ((1<=r & r<=x/2) => z=3*x+4*(r-1)) & ((x/2 < r & r<=x) => z=4*x+2*(r-1))":

eval test0 "An,y0,y1,y2,y4 ($tmsc(8*n,y0) & $tmsc(8*n+1,y1) & 
   $tmsc(8*n+2,y2) & $tmsc(8*n+4,y4)) => $tmsc(16*n,(y0+3*y2)-(y1+y4))":

eval test1 "An,y2,y4 ($tmsc(8*n+2,y2) & $tmsc(8*n+4,y4)) 
   => $tmsc(16*n+1,(3*y2)-y4)":

eval test2 "An,y1,y2,y4 ($tmsc(8*n+1,y1) & $tmsc(8*n+2,y2) & 
   $tmsc(8*n+4,y4)) => $tmsc(16*n+2,(4*y2)-(y1+y4))":

eval test3 "An,y1,y2,y4 ($tmsc(8*n+1,y1) & $tmsc(8*n+2,y2) & 
   $tmsc(8*n+4,y4)) => $tmsc(16*n+3,(5*y2)-(2*y1+y4))":

eval test4 "An,y1,y2 ($tmsc(8*n+1,y1) & $tmsc(8*n+2,y2)) 
   => $tmsc(16*n+4,(3*y2)-y1)":

eval test5 "An,y1,y2 ($tmsc(8*n+1,y1) & $tmsc(8*n+2,y2)) 
   => $tmsc(16*n+5,(4*y2)-2*y1)":

eval test6 "An,y1,y2,y4 ($tmsc(8*n+1,y1) & $tmsc(8*n+2,y2) & 
   $tmsc(8*n+4,y4)) => $tmsc(16*n+6,(2*y2+y4)-y1)":

eval test7 "An,y4 $tmsc(8*n+4,y4) => $tmsc(16*n+7,2*y4)":

eval test8 "An,y1,y2,y4 ($tmsc(8*n+1,y1) & $tmsc(8*n+2,y2) & 
   $tmsc(8*n+4,y4)) => $tmsc(16*n+8,(y2+2*y4)-y1)":

eval test9 "An,y1,y2,y4 ($tmsc(8*n+1,y1) & $tmsc(8*n+2,y2) & 
   $tmsc(8*n+4,y4)) => $tmsc(16*n+9,(2*y2+2*y4)-2*y1)":

eval test10 "An,y1,y2,y4,y6 ($tmsc(8*n+1,y1) & $tmsc(8*n+2,y2) & 
   $tmsc(8*n+4,y4) & $tmsc(8*n+6,y6)) => $tmsc(16*n+10,(y2+y4+y6)-y1)":

eval test11 "An,y6 $tmsc(8*n+6,y6) => $tmsc(16*n+11,2*y6)":

eval test12 "An,y1,y2,y4,y6 ($tmsc(8*n+1,y1) & $tmsc(8*n+2,y2) & 
   $tmsc(8*n+4,y4) & $tmsc(8*n+6,y6)) => $tmsc(16*n+12,(y1+3*y6)-(y2+y4))":

eval test13 "An,y1,y2,y4,y6 ($tmsc(8*n+1,y1) & $tmsc(8*n+2,y2) & 
   $tmsc(8*n+4,y4) & $tmsc(8*n+6,y6)) => 
   $tmsc(16*n+13,(2*y1+4*y6)-(2*y2+2*y4))":

eval test14 "An,y1,y4,y6 ($tmsc(8*n+1,y1) & $tmsc(8*n+4,y4) & 
   $tmsc(8*n+6,y6)) => $tmsc(16*n+14,(y1+5*y6)-4*y4)":

eval test15 "An,y2,y4,y6 ($tmsc(8*n+2,y2) & $tmsc(8*n+4,y4) & 
   $tmsc(8*n+6,y6)) => $tmsc(16*n+15,(2*y2+6*y6)-6*y4)":
\end{verbatim}

\end{document}